\documentclass[12pt]{amsart}
\usepackage{fullpage, bbm}
\usepackage{graphicx}

\usepackage{amsfonts}
\usepackage{amssymb}
\usepackage{amsmath}
\usepackage{amsthm}

\newtheorem{theorem}{Theorem}[section]

\newtheorem{corollary}[theorem]{Corollary}
\newtheorem{proposition}[theorem]{Proposition}
\newtheorem{definition}[theorem]{Definition}

\theoremstyle{remark}\newtheorem{remark}[theorem]{Remark}

\newcommand{\EE}{\mathbb{E} }
\newcommand{\one}{\mathbbm{1} }
\newcommand{\QQ}{\mathbb{Q} }
\newcommand{\PP}{\mathbb{P} }
\newcommand{\RR}{\mathbb{R} }

\newcommand{\ff}{\mathcal{F} }

\begin{document}

\author{Michael R. Tehranchi \\
University of Cambridge }
\address{Statistical Laboratory\\
Centre for Mathematical Sciences\\
Wilberforce Road\\
Cambridge CB3 0WB\\
UK}
\email{m.tehranchi@statslab.cam.ac.uk}

\title{Arbitrage theory without a num\'eraire}

\date{\today}

\begin{abstract} This note develops an arbitrage theory
for a discrete-time market model without the assumption of
the existence of a num\'eraire asset.  Fundamental theorems of asset pricing  are stated
and proven in this context. The distinction between the notions of investment-consumption arbitrage
and pure-investment arbitrage provide a discrete-time analogue of the distinction
between the notions of absolute arbitrage and relative arbitrage in the 
continuous-time theory.  Applications to the modelling of bubbles is discussed.
\end{abstract}

\maketitle
 
\section{Introduction}

In most accounts of arbitrage
 theory, the concept of an equivalent martingale measure takes centre stage.  
Indeed, the discrete-time fundamental theorem of asset pricing,
first proven by Harrison \& Kreps \cite{HK} for models with finite sample spaces
and by Dalang, Morton \& Willinger \cite{DMW} for general models, 
says that there is no arbitrage 
 if and only if there exists an equivalent martingale measure.

An equivalent martingale measure is defined in terms of a given num\'eraire asset.
Recall a num\'eraire is an asset, or more generally a portfolio, whose
price is strictly positive at all times with probability one.  Associated
to a given equivalent martingale measure 
is a positive adapted process.  This process is often called a martingale deflator, but
is also known
as a pricing kernel, a stochastic discount factor or a state price density.  It seems
that martingale deflators feature less prominently in
the financial mathematics literature, although they are, in a sense, 
more fundamental.   Indeed, they are the natural dual variables
for an investor's optimal investment problem and have the economic
interpretation as the sensitivity of the maximised expected utility with
respect to the current level of wealth.  Furthermore, 
 unlike the concept of an equivalent martingale measure,
 the concept of a martingale deflator is defined in a completely
 num\'eraire-independent manner.

Note that in order to define an equivalent martingale measure, it is 
necessary to assume that at least one num\'eraire exists.  This assumption
is ubiquitous in the financial mathematics literature, but as we will see,
it is not strictly necessary.  In particular, in this note, we consider a discrete-time
arbitrage theory without the assumption of the existence
of a num\'eraire, and we will see that  fundamental theorems of asset pricing can be formulated in 
this setting.  From an aesthetic, or possibly pedantic, perspective, 
 we dispense with a mathematically
unnecessary assumption and rephrase the characterisation of an
arbitrage free market in terms of the more fundamental notion of a martingale deflator. 
  However, there are other reasons to weaken the assumptions of the theorem.

While it does not seem that the assumption of the existence of a num\'eraire 
is controversial, 
it is not entirely innocent either.  Indeed, 
there is growing interest in robust arbitrage theory, where the assumption
that there exists a single dominating measure is dropped.  See, for 
instance,  recent papers of  Bouchard \& Nutz \cite{BN} and  Burzoni,   Frittelli \&   Maggis
\cite{BFM} for robust versions of the 
fundamental theorem of asset pricing in discrete time.  From the perspective of
robust finance, the assumption of the existence of a num\'eraire seems rather strong.
Indeed, when dealing with a family of possibly singular measures, insisting that 
there is an asset with strictly positive price almost surely under all such measures 
might be asking too much.

One benefit of our more general treatment of 
arbitrage theory is that it provides some analogues in the discrete-time theory
that previously have been considered generally to be only continuous-time phenomena.
In particular, we will see that the distinction between the notions of investment-consumption arbitrage
and pure-investment arbitrage provide a discrete-time analogue of the distinction
between the notions of absolute arbitrage and relative arbitrage in the 
continuous-time theory.  
In particular, when the market does not admit a num\'eraire, it
is possible for there to exist a price bubble in discrete-time in the same spirit as
the continuous-time notion of bubble popularised by Cox \& Hobson \cite{CH}
and Protter \cite{protter}.

The remainder of the paper is arranged as follows.  In section \ref{se:invcon}
we introduce the notation and basic definitions as well as the 
main results of this paper: a characterisation of no-arbitrage
in a discrete-time model without a num\'eraire.  We also characterise
the minimal superreplication cost of a contingent claim in this context,
and show that the martingale deflator serve as the dual variables for 
optimal investment problems, even when no num\'eraire is assumed to exist.
In section \ref{se:pure} we introduce the notion of a pure-investment strategy,
and  characterise contingent claims which can be replicated by such strategies.
We recall the notion of a complete market, and show among other things, that
if the market is complete and arbitrage free, then  necessarily there exists
a risk-free num\'eraire.    
In addition, we recapture the classical no-arbitrage results when we assume that a num\'eraire
exists.     In section \ref{se:abs-rel} we explore other notions of arbitrage and
show that they are not equivalent in general. We also
discuss how these differences could be used to define a bubble
in a discrete-time model, in analogy of a popular definition of a bubble
in continuous-time models as recalled in section  \ref{se:relative}.
In section \ref{se:proofs} we present the proofs of the main results, 
along with the key economic insight arising from the optimal investment problem.
The ideas here originated in Rogers's \cite{rogers} proof of the Dalang--Morton--Willinger
theorem.  In particular,  we present
the full details since they are rather 
easy and probabilistic, and do not rely on any knowledge of convex analysis
or separation theorems in function spaces.  Furthermore, we use this
utility maximisation-based framework to give a novel proof of the
 characterisation of super-replicable
claims.  
Finally, in section \ref{se:tech} we include a few technical lemmas 
regarding measurability and discrete-time local martingales.  We note here
that we do not rely on any general measurable selection theorems, preferring
a hands-on treatment.

\section{Investment-consumption strategies}\label{se:invcon} 
We consider a general frictionless market model where there are $n$ assets.
We let $P^i_t$  denote the price of asset $i$ at time $t$,
where we make the simplifying assumption that no asset pays a dividend.
We use the notation $P_t = (P_t^1, \ldots, P_t^n)$ to denote the vector of asset
prices, and we model these prices as a $n$-dimensional adapted stochastic process
$P=(P_t)_{t \ge 0}$ defined on some probability space $(\Omega, \ff_{\infty}, \PP)$
with filtration $\ff = (\ff_t)_{t \ge 0}$.  Time is discrete
so the notation $t \ge 0$ means $t \in \{0, 1, 2, \ldots \}$.  We will also
use the notation
$a \cdot b = \sum_{i=1}^n a^{i}b^{i}$ to denote the usual Euclidean inner  product in $\RR^{n}$.

\begin{remark}
Note that we do not make any assumption about the sign of any of the
random variables $P_t^i$. When $P_t^i = 0$ the asset is worthless  and 
when $P^i_t < 0$ the asset is actually a liability.  This flexibility
allows us to handle claims, such as forward contracts, whose payouts can
be positive, zero or negative.  

However, in most presentations of discrete-time arbitrage theory, one assumes that
there is at least one asset with strictly positive price, 
that is, that there is at least one $i$
such that $P_t^i > 0$ almost surely for all $t \ge 0$. 
Such an asset is called a num\'eraire because the prices of the other assets
can be written as multiples of the price of the num\'eraire.
Usually, the assumption of the existence of a num\'eraire
	is treated as natural because of the following examples.

In the first example, we make the trite observation that prices, by definition, must be
 denominated in some currency.  Therefore, one  may
choose the num\'eraire asset   to be the currency
itself, in which case $P^{\mathrm{num}}_t = 1$   for all $t$. 

In the second example, one   supposes that there is a  central 
bank which issues bonds at each time $t$ maturing
at time $t+1$.  Furthermore, the central
bank is assumed to be perfectly risk-free and the interest rate
 is assumed to be a constant $r > 0$.   The value at time $t$
of a money market account from the initial investment of one
unit of currency is then $(1+r)^t$.  In such a market set-up, one may
assign  the num\'eraire asset the price $P^{\mathrm{num}}_t  = (1+r)^t$ for all $t$.   

Note that in second example, we usually \textit{do not} include the currency
itself as an asset in our market model.  Indeed, otherwise, there would be the trivial
arbitrage of holding a short position in the currency and a long position
in the money market.  Therefore, in such models, we think of the currency as a
unit of account and means of exchange, but the money market plays the 
role of a store of value.  

In the models treated in this paper, we simply remove the assumption
that there is some store of value.  This allows us to model, for instance, an economy 
experiencing hyperinflation  where every asset (and even the currency itself) 
suffer a non-zero probability of becoming worthless in the future. 
Because there is no store of value, we think of prices as being
denominated in some perishable consumption good.  Of course, 
we must then allow agents the possibility
to consume this good.  A related continuous-time market
model with  hyperinflation is considered in the paper of Carr, Fisher \& Ruf \cite{CFR}.
\end{remark}

To the market described by the  process $P$, we now
introduce an investor.
Suppose that $H^i_t$ is the number of shares
of asset $i$ held during the interval $(t-1, t]$ of time. 
We will allow $H^i_t$ to be either positive, negative, or zero with the 
	interpretation that if $H_t^i > 0$ the investor is long asset $i$ 
	and if $H_t^i < 0$ the investor is short the asset. Also, we do not demand
	that the $H^i_t$ are integers.  As usual, we  introduce a self-financing
	constraint on the possible dynamics of the  $n$-dimensional process
$H = (H^1_t, \ldots, H^n_t)_{t \ge 1}$.

\begin{definition}
An investment-consumption strategy 
is an $n$-dimensional  predictable process $H$
satisfying the self-financing condition 
\begin{align*}
H_{t} \cdot P_{t} &\ge  H_{t+1} \cdot P_{t} \mbox{ almost surely for all } t \ge 1.
\end{align*}
\end{definition}

 \begin{remark} 
 The
idea is that  the investor brings the initial capital $X_0$ to 
the market.  He then consumes a non-negative amount $C_0$, and 
invests the remainder into the market by choosing a vector of 
portfolio weights $H_1 \in \RR^n$ such that $H_1 \cdot P_0 = X_0-C_0$.

At each future time $t \ge 1$, the investor's pre-consumption
wealth is just the market value  $X_t = H_t \cdot P_t$ of his current holdings.
He again chooses a non-negative amount $C_t$ to consume,
and use the post-consumption wealth $X_t- C_t = H_{t+1} \cdot P_t$ to rebalance his 
portfolio to be held until time $t+1$ when the market clock ticks again.

The assumption that the strategy $H$ is predictable models the fact that the
investor is not clairvoyant.
Note that while a strategy $H$ is predictable by definition, we have chosen the convention
that the consumption stream $C$ is merely adapted.
\end{remark}

\begin{remark}  The `free-disposal' assumption, or allowing agents to
`throw away money', has long been a part of the  classical  arbitrage theory
in which the existence of a num\'eraire asset is assumed.
Indeed, in infinite-horizon discrete-time models or in continuous-time models,
this assumption is indispensable to the
 formulation of a dual characterisation of no-arbitrage. 
 See for instance the paper of Schachermayer
\cite{Sch}.  We will see that even for a finite-horizon discrete-time model,
such a free-disposal assumption is needed when there is no num\'eraire.
\end{remark}

Supposing that an investor in this market has a preference relation over the set of 
investment-consumption strategies, his goal then is to find the best strategy
given his budget constraint. 
 To fix ideas, suppose that his preference
has a numerical representation, so that strategy $H$ is preferred to strategy $H'$
if and only if $U(H) > U(H')$ where $U$
has the additive expected utility form
\begin{equation}\label{eq:U}
U(H) =  \EE \left[ u(C_0, \ldots, C_T) \right]
\end{equation}
where $X_0$ is his initial wealth and $T > 0$ is a fixed, non-random time horizon, and
where $C_0 = X_0-H_1 \cdot P_0$ and $C_t = (H_t-H_{t+1})\cdot P_t$ for $t \ge 1$ is
the investor's consumption.
Assuming that he is not permitted to have negative wealth after time $T$, 
the investor's problem is to maximise $U(H)$ subject to the budget constraint
$H_1 \cdot P_0 \le X_0$ and the transversality condition $H_{T+1} = 0$.  For future
reference, we will let
\begin{equation}\label{eq:feas}
\mathcal{H}_{X_0,T} = \{ H: \mbox{ self-financing, }  H_1 \cdot P_0 \le X_0, H_{T+1} = 0 \}
\end{equation}
be the set of feasible solutions to this problem.
 
The notions of arbitrage are intimately related to whether this
optimal investment problem has a solution. Furthermore, we
 will see that martingale deflators are the dual variables for this optimisation
problem.  The fundamental theorems stated below establish the connection between
the absence of  arbitrage and the existence of martingale deflators.

\subsection{Martingale deflators and optimal investment}\label{se:optinv}
We now come to the definition of a martingale deflator.
For technical reasons, it will also be useful to introduce the related concept
 of a local martingale deflator.

\begin{definition}
 A (local) martingale deflator is a strictly positive adapted process $Y =(Y_t)_{t \ge 0}$
such that the $n$-dimensional process 
$
PY = \left(P_t Y_t\right)_{t \ge 0}
$
is a (local) martingale. 
\end{definition}

\begin{remark}
A concept  very closely related to that of a martingale deflator is that of a equivalent
martingale measure, whose definition is recalled in section \ref{se:num} below.
While the definition of an equivalent martingale measure 
is highly asymmetric, in the sense that it gives a distinguished
role to one asset among the $n$ total assets, note that the definition of a martingale deflator
is perfectly symmetric in the sense that all assets are treated equally.
\end{remark}

The key result underpinning many of the arguments to come is the following proposition:
\begin{proposition}\label{th:mart}
Let $Y$ be a local martingale deflator and $H$ an investment-consumption strategy.
Fix $X_0 \le H_1 \cdot P_0$ and let $C_t = X_t - H_{t+1} \cdot P_t$ for $t \ge 0$,
where $X_t = H_t \cdot P_t$ for $t \ge 1$.

The process $M$ defined by
$$
M_t = X_t Y_t + \sum_{s=0}^{t-1} C_s Y_s
$$
is a local martingale. 
In particular, if $H_{T+1}=0$ for some non-random $T > 0$, then
$$
\EE \left( \sum_{s=0}^{T} C_s Y_s \right) = X_0 Y_0.
$$
\end{proposition}

\begin{proof}
Note that by rearranging the sum, we have the identity
$$
M_t = X_0 Y_0 + \sum_{s=1}^t  H_s \cdot ( P_s Y_s - P_{s-1} Y_{s-1}).
$$
If $Y$ is a local martingale deflator, then $M$ is a local martingale
by Proposition \ref{th:mart-trans}. 

For the second claim, note that if $H$ is a self-financing
investment-consumption strategy with $H_{T+1} = 0$, then 
$C_T = X_T$ and hence
\begin{align*}
M_T &=  \sum_{s=1}^{T} C_s Y_s \\
& \ge 0.
\end{align*}  
By Proposition \ref{th:non-neg}, the process $(M_t)_{0 \le t \le T}$ is
a true martingale and hence by the optional sampling theorem we have
\begin{align*}
 \EE( M_T ) & = M_0 \\
& = X_0 Y_0.
\end{align*}
\end{proof}

We now turn our attention to the  utility
optimisation problem described above.  In particular, we will see
that martingale deflators play the role of a dual variable or Lagrange multiplier.  The
following proposition is not especially new, but it does again
highlight the fact that utility maximisation theory does not 
depend on the existence of a num\'eraire.
In particular, the following theorem gives us the interpretation of a martingale deflator as the 
marginal utility of an optimal consumption stream.
 
\begin{theorem}\label{th:util}
Let the set of investment-consumption strategies $\mathcal{H}_{X_0,T}$
be defined by equation \eqref{eq:feas}, and let $U(H)$ be defined by equation \eqref{eq:U} 
for  $H \in \mathcal{H}_{X_0,T}$   such that the expectation
is well defined.  Furthermore, suppose that the utility function $u$ is   convex and differentiable
such that $c_t \mapsto u(c_0, \ldots, c_T)$ is strictly increasing  
for each $t$.  

If there is a local martingale deflator $Y$ and a feasible strategy $H^* \in \mathcal{H}_{X_0,T}$
 such that
$$
Y_t = \EE \left( 
\frac{\partial}{\partial c_t} u(C^*_0 , \ldots, C^*_T ) | \ff_t \right) 
$$
in the sense of generalised conditional expectation (recalled in section \ref{se:tech})
where 
$C^*_0 = X_0-H^*_1 \cdot P_0$ and $C^*_t = (H^*_t-H^*_{t+1})\cdot P_t$ for $t \ge 1$,
then $H^*$ is optimal in the sense that
$$
U(H^*) \ge U(H)
$$
for all feasible $H$.
\end{theorem}

\begin{proof}
Let $H$ be feasible with corresponding consumption stream $C$.  
First observe that since $C_t \ge 0$ and $\frac{\partial}{\partial c_t} u(C^*) > 0$ we have
by the slot property (Proposition \ref{th:slot} below) 
$$
\EE\left( C_t \frac{\partial}{\partial c_t} u(C^*) | \ff_t \right) = C_t Y_t
$$
and by the tower property (Proposition \ref{th:tower}) that
\begin{align*}
\EE\left(  \sum_{t=0}^{T} C_t \frac{\partial}{\partial c_t} u(C^*) \right)
& = \EE \left(\sum_{s=0}^{T} C_s Y_s \right) \\
& = X_0
\end{align*}
where we have used Proposition \ref{th:mart} between the second and third line. 
 By the convexity of the function $u$ we have
$$
u(C) \le u(C^*) +  \sum_{t=0}^{T} (C_t-C^*_t) \frac{\partial}{\partial c_t} u(C^*)
$$
almost surely.  The conclusion follows from taking the expectation of 
both sides and applying the first observation to cancel the sum.
\end{proof}

\subsection{Arbitrage and the first fundamental theorem}\label{se:1ftap}

We introduce the following definition:

\begin{definition} An investment-consumption   arbitrage is a predictable process
$H$ such that there exists a non-random  time horizon $T > 0$ with the properties that
\begin{itemize}
 \item $H \in \mathcal{H}_{0, T}$, 
\item $\PP\left(  (H_t - H_{t+1})\cdot P_t  > 0 \mbox{ for some } 0 \le t \le T \right) > 0$
where $H_0= 0$.
\end{itemize}
\end{definition}

Suppose that $H^{\mathrm{arb}}$ is an arbitrage according to the above definition.  
If $H \in \mathcal{H}_{X_0,T}$ for some initial wealth $X_0$ and time horizon
 $T$,   then $H+H^{\mathrm{arb}} \in  \mathcal{H}_{X_0,T}$.
Furthermore, if the functions $c_t \mapsto u(c_0, \ldots, c_T)$ are strictly increasing   
 then $U(H+H^{\mathrm{arb}}) > U(H)$,   and hence the strategy
$H+H^{\mathrm{arb}}$ is strictly preferred to $H$. 
In particular, the optimal investment problem 
cannot have a solution.  
 In section \ref{se:proofs} we show that a
certain converse is true: if there is no arbitrage, then it is possible to formulate an 
optimal investment problem that has a maximiser.

We now come to our version of the first fundamental theorem of asset pricing
for investment-consumption strategies.

\begin{theorem}\label{th:1FTAP} The following are equivalent:
\begin{enumerate}
\item The market has no investment-consumption arbitrage.
\item There exists a local martingale deflator.
\item There exists a martingale deflator.
\item For every non-random $T > 0$ and every positive adapted process $(\eta_t)_{0 \le t \le T}$,
there exists a martingale deflator $(Y_t)_{0 \le t \le T}$ such that $Y_t \le \eta_t$ almost
surely for all $0 \le t \le T$.
\end{enumerate}
\end{theorem}

The equivalence of  (1) and (3) above is the real punchline of the story.  Condition (2) 
is some what technical, but is useful since it is easier to check than condition (3).
  Condition (4)   will prove very useful in the next
section  since it implies that for any $\ff_T$-measurable random variable $\xi_T$ 
there exists martingale deflator $Y$ such that $\xi_T Y_T$ is integrable.
Notice that although it is true that (4) implies  (3), the argument is not as
trivial as it might first seem, since condition (4) holds for each fixed time horizon $T$, while condition
(3) says that $(P_t Y_t)_{t \ge 0}$ is a martingale over an infinite horizon.

We now prove the equivalence of conditions (2) and (3).  
\begin{proof}  [Proof of (2) $\Leftrightarrow$ (3) of Theorem \ref{th:1FTAP}]
Since a martingale deflator is also 
a local martingale deflator, we need only prove (2) $\Rightarrow$ (3).  

Let $Y$ be a local martingale deflator, so that $PY$ is a local martingale.
Note that if we were to assume that each asset price is non-negative, 
so that $P_t^i \ge 0$ almost surely
for all $1 \le i \le n$ and $t \ge 0$, we could invoke Proposition \ref{th:non-neg}
to conclude that $PY$ is a true martingale and, hence, that $Y$ is a 
true martingale deflator.  In the general case, we appeal to 
Kabanov's theorem \cite{Kabanov}, quoted as Theorem \ref{th:kabanov} below, which
says that there exists an equivalent measure $\QQ$
such that $PY$ is a true martingale under $\QQ$.    Letting
$$
\hat Y_t = Y_t \EE^{\PP} \left( \frac{d\QQ}{d\PP} | \ff_t \right),
$$
we see that $P\hat Y$ is a true martingale under $\PP$, and hence
$\hat Y$ is a true martingale deflator.
\end{proof}

We now prove that the existence of a local martingale deflator implies the absence of investment-consumption
arbitrage. This is a well-known argument, but we include it here
for completeness.

\begin{proof}[Proof of (2) $\Rightarrow$ (1)  and (4) $\Rightarrow$ (1) of Theorem \ref{th:1FTAP}]
Fix $T > 0$, and fix a strategy $H \in \mathcal{H}_{0,T}$. 
If $(Y_t)_{0 \le t \le T}$ is a local martingale deflator then, 
 Proposition \ref{th:mart} implies that 
$$
\EE \left[ \sum_{s=0}^{T} C_s Y_s \right] =   0
$$
where   $C_t = (H_t-H_{t+1})\cdot P_t$.
Since $Y_t > 0$ and $C_t \ge 0$ almost surely for all $t \ge 0$, the 
pigeon-hole principle and the equality
above imply $C_t = 0$ almost surely for all $0 \le t \le T$.  Hence $H$
is not an investment-consumption arbitrage.
\end{proof}

The proofs of  (1) $\Rightarrow$ (2), that   no-arbitrage implies both the existence of a local martingale deflator, 
as well as of (1) $\Rightarrow$ (4), that the existence over any finite time horizon of suitably bounded martingale
deflator, are more technical and deferred to section \ref{se:proofs}.

\begin{remark}
Note that the one period case of the (1) $\Rightarrow$ (4) implication of our num\'eraire-free
fundamental theorem is
\begin{proposition}\label{th:numfree1}  Let $p \in \RR^n$ be constant
and $P$ be a random vector valued in $\RR^n$ 
with the property that
$$
h \cdot p  \le 0 \le h \cdot P \mbox{ implies } h \cdot p = 0 = h \cdot P
$$
Then there exists a bounded $Y > 0$ such that $\EE(PY) = p$.
\end{proposition}

It should be noted that this case can be derived from the classical num\'eraire-dependent one-period fundamental theorem:
\begin{proposition}\label{th:num1}
Let $X$ be a random vector valued in $\RR^d$ with the property that 
$$
a \cdot X \ge 0 \mbox{ implies } a \cdot X = 0.
$$
Then there exists a bounded $Z > 0$ such that $\EE(XZ) = 0$.
\end{proposition}

\begin{proof}[Proof of Proposition \ref{th:numfree1}] Given
the probability space $(\Omega, \ff, \PP)$, let $\Delta$ be some
new state not in $\Omega$, and enlarge
the space by letting $\hat \Omega = \Omega \cup \{ \Delta \}$,
$\hat \ff = \sigma( \ff, \{\Delta\} )$ and let $\hat \PP$ be the measure on $\hat \ff$
such that
$$
\hat{\PP}(A) = \frac{1}{2} \PP(A), \ \ \mbox{ if } A \in \ff
$$
and 
$$
\hat{\PP}\{ \Delta\} = \frac{1}{2}.
$$
Finally, define
$$
X(\hat \omega) = \left\{ \begin{array}{ll}
 P(\hat \omega) & \mbox{ if } \hat \omega \in \Omega \\
 - p & \mbox{ if } \hat \omega = \Delta
\end{array} \right.
$$
Note that if $a \cdot X \ge 0 \ \ \hat \PP$ a.s. then 
$a \cdot P \ge 0 \ \ \PP$-a.s and $a \cdot p \le 0$.
By assumption then, $a \cdot P = 0 \ \ \PP$-a.s and $a \cdot p = 0$ 
which again implies $a \cdot X = 0 \ \ \hat \PP$ a.s.
Hence by Proposition \ref{th:num1} there exists a bounded $Z$ on $\hat\Omega$ such that
$$
\hat \EE( Z X ) = 0.
$$
Let $Y$ be the restriction of $Z$ to $\Omega$ and $y=Z(\Delta)$.
The above equation becomes
$
\EE( Y P) = y p
$
as desired.
\end{proof}
Unfortunately, it is difficult to adapt this proof to the multi-period case. 
\end{remark}

\subsection{Super-replication}
We now turn to the 
dual characterisation of super-replication.  It is a classical
theorem of the field, but we include it here since it might be surprising to know that it holds without the
assumption of the existence of a num\'eraire.

\begin{theorem}\label{th:superrep}   Suppose there exists at least one martingale deflator. 
Let $\xi$ be an adapted process such that $\xi Y$ is a supermartingale for all  martingale deflators $Y$ 
for which $\xi Y$ is an integrable process. Then 
there exists an investment-consumption strategy $H$  such that 
\begin{align*}
H_1 \cdot P_0 &\le \xi_0 \\
H_{t+1} \cdot P_t & \le \xi_t \le H_t \cdot P_t
\end{align*}
\end{theorem}

The proof is deferred to section \ref{se:proofs}.

Note that given a non-negative $\ff_T$-measurable random variable $X_T$, we can
find the smallest process $(X_t)_{0 \le t \le T}$ such that $X Y$ is a supermartingale for all $Y$ by
\begin{align*} 
X_t &= \mathrm{ess \ sup}\left\{  \frac{1}{Y_t} \EE(X_T Y_T  | \ff_t ): Y \mbox{ a martingale deflator such that }
X_T Y_T \mbox{ is integrable } \right\} 
\end{align*}
Theorem \ref{th:superrep} says that there is a strategy $H$ such that
$H_T \cdot P_T \ge X_T$ almost surely, such that $H_1 \cdot P_0 \le X_0$.  
 In other words, $X_0$ bounds the initial cost of super-replicating the 
contingent claim with payout $\xi_T$.

\section{Pure-investment strategies}\label{se:pure}
In constrast the investment-consumption strategies
studied above, we now introduce the notion of a pure-investment strategy.

\begin{definition}
A strategy $H$ is called a  pure investment strategy if 
$$
  (H_t - H_{t+1}) \cdot P_{t} = 0 \mbox{ almost surely for all } t \ge 1.
$$
For pure-investment strategies, we will make the convention that $H_0 = H_1$.
\end{definition}

\subsection{Replication and the second fundamental theorem}

We are already prepared to characterise
contingent claims which can be attained by pure investment.  Again, the result is not
especially new, but it is interesting to see that it holds without the assumption
of the existence of  a num\'eraire.

\begin{theorem}\label{th:rep}    Suppose there exists at least one martingale deflator. 
Let $\xi$ be an adapted process such that $\xi Y$ is a martingale for all martingale deflators $Y$ 
such that $\xi Y$ is an integrable process. Then 
there exists a pure-investment strategy $H$ such that $H_t \cdot P_t = \xi_t$ almost surely for all $t \ge 0$.

In particular, if $\xi_T$ is an $\ff_T$-measurable random variable
such that there is a constant $\xi_0$ such that $\EE( \xi_T Y_T) = \xi_0 Y_0$ for 
all martingale deflators $Y$
for which $ \xi_T Y_T$ is integrable, then there is a pure-investment strategy $H$ such that $H_T \cdot P_T =\xi_T$.
\end{theorem}

\begin{proof}
If $\xi Y$ is a martingale for all suitably integrable
 martingale deflators $Y$, then 
by Theorem  \ref{th:superrep}, there exists an investment-consumption
strategy $H$ such that $\xi_t \le H_t \cdot P_t$ almost
surely for all $t \ge 1$ and that $\xi_0 \ge H_1 \cdot P_0$.
Hence, by setting $X_0 = \xi_0$ and  $X_t = H_t \cdot P_t$ for $t \ge 1$,
we see that $C_t = X_t - H_{t+1} \cdot P_t$ is non-negative for all $t \ge 0$.
 Now fix one such $Y$ and let
$$
M_t  =  X_t Y_t + \sum_{s=0}^{t-1} C_s Y_s.
$$ 
  Note that
$M$ is a local martingale by Proposition
\ref{th:mart}.  Hence the process $M- \xi Y$ is a local
martingale by linearity.  But note that
since $X_t \ge \xi_t$ we have that 
$$
M_t - \xi_t Y_t   \ge 0
$$
and hence $M- \xi Y$ is a true martingale by
Proposition \ref{th:non-neg}.  However, 
$$
\EE(M_t - \xi_t Y_t) = (X_0- \xi_0) Y_0 = 0 
$$
and hence $M_t = \xi_t Y_t$ for all $t \ge 0$ which
implies $X_t = \xi_t$ for all $t \ge 0$ as desired.  

Now suppose  $\xi_T$ is a given $\ff_T$-measurable random variable
such  $\EE( \xi_T Y_T) = \xi_0 Y_0$ for all martingale deflators $Y$
and a given real number $\xi_0$.  
Let 
$$
D_t =  \underset{Y}{\mathrm{ess \ sup}} \ \frac{1}{Y_t} \EE(\xi_T Y_T  | \ff_t )
- \underset{Y}{\mathrm{ess \ inf}} \ \frac{1}{Y_t} \EE(\xi_T Y_T  | \ff_t )
$$
where the essential supremum and infimum are over all suitably integrable
martingale deflators $Y$.  Note that $DY$ is submartingale for all $Y$,
but $D_T = 0 = D_0$.  Hence $D_t = 0$ for all $0 \le t \le T$ and so
we can set
$$
\xi_t = \mathrm{ess \ sup}_Y \frac{1}{Y_t} \EE(\xi_T Y_T  | \ff_t )
= \mathrm{ess \ inf}_Y \frac{1}{Y_t} \EE(\xi_T Y_T  | \ff_t )
$$
for $0 < t < T$ and apply the previous result.
\end{proof}

\begin{remark} Another proof of this theorem is given in \cite{T2010}
in the case where the market has a num\'eraire asset.
\end{remark}

We now recall a definition.
\begin{definition}  A market model is complete if for every $T > 0$
and $\ff_T$-measurable random variable $\xi_T$ there exist a 
pure-investment strategy $H$ such that
$$
H_T \cdot P_T = \xi_T.
$$
\end{definition}

In this framework we can state the seond fundamental theorem of asset pricing.
\begin{theorem}   Suppose the market has no arbitrage.  The market
is complete if and only if there is exactly one martingale deflator such that $Y_0 = 1$.
\end{theorem}

\begin{proof}
 Suppose that there is a unique martingale deflator such that $Y_0 = 1$.
Let $\xi_T$ be any $\ff_T$-measurable random variable.  By implication (3) of Theorem \ref{th:1FTAP}
we can suppose that $\xi_T Y_T$ is integrable.  Let
$$
\xi_t = \frac{1}{Y_t} \EE( Y_T \xi_T | \ff_t).
$$
Note that $\xi Y$ is a martingale.  Hence $\xi Y'$ is
also a martingale for any martingale deflator $Y'$ since $Y'=Y'_0 Y$.
  By Theorem \ref{th:rep} there exists a pure-investment
strategy such that $H \cdot P = \xi$.  Hence the market is complete.

Conversely, suppose that the market is complete.  Let $Y$ and $Y'$ be
martingale deflators such that $Y_0=Y_0'=1$.  Fix a $T> 0$.
By completeness there exists
a pure-investment strategy $H$ such that
$$
H_T \cdot  P_T =( Y_T - Y_T' ) Z 
$$
where $Z = \frac{1}{(Y_T + Y_T')^2}$.  (The factor $Z$ will be used 
to insure integrability later.)

Since $H \cdot P Y$ is a local martingale by Proposition \ref{th:mart}
with non-negative value at time $T$, it is a true martingale by
Proposition \ref{th:non-neg}. 
In particular,
$$
H_0 \cdot P_0 =  \EE [( Y_T - Y_T' ) Z  Y_T ].
$$
By the same argument with $Y'$ we have
$$
H_0 \cdot P_0 =  \EE [( Y_T - Y_T' ) Z  Y_T' ].
$$
Subtracting yields
$$
\EE[ ( Y_T - Y_T' )^2 Z   ]  = 0
$$
so by the pigeon-hole principle we have $\PP( Y_T = Y_T') = 1$ as desired.
\end{proof}

In discrete-time models, complete markets
have even more structure:

\begin{theorem}  Suppose the market model with $n$ assets is complete.
For each $t \ge 0$, every $\ff_t$-measurable partition of the sample space $\Omega$
has  no more than $n^t$  events of positive probability.  In particular, the $n$-dimensional random vector $P_t$ takes  
values in a set of at most $n^t$ elements.
\end{theorem}

\begin{proof} Fix $t \ge 1$.  Let $A_1, \ldots, A_p$ be
a maximal partition of $\Omega$ into disjoint $\ff_{t-1}$-measurable events with $\PP(A_i) > 0$ for all $i$.
Similarly, let $B_1, \ldots, B_q$ be a   partition into non-null $\ff_t$-measurable events.  We
will show that $q \le n p$.   The result will follow from induction.

First note that the set $\{ \one_{B_1}, \ldots, \one_{B_q} \}$ of random variables
is  linearly independent,
and in particular, 
the dimension of its span is exactly $q$.
Assuming that the market is complete, each of the $\one_{B_i}$ can be replicated by a pure-investment strategy.  Hence
$$
\mathrm{span}\{ \one_{B_1}, \ldots, \one_{B_q} \}   \subseteq \{ H\cdot P_t: H \mbox{ is $\ff_{t-1}$-meas. } \}.
$$
We need only show that the dimension of the space on the right is at most $n p$.

Now note that if a random vector $H$ is $\ff_{t-1}$-measurable, then
it takes exactly one value on each of the $A_j$'s for a total of at most $p$ values
$h_1, \ldots, h_p$.  Hence
\begin{align*}
\{ H\cdot P_t: H \mbox{ is $\ff_{t-1}$-meas. } \} & = 
\{ h_1 \cdot P_t \one_{A_1} + \ldots + h_p \cdot P_t \one_{A_p}: h_1, \ldots, h_p \in \RR^n \} \\
&= \mathrm{span}\{ P_t^i\one_{A_j}: 1 \le i \le n, 1 \le j \le p \},
\end{align*}
concluding the argument.
\end{proof}

The final result shows that in a complete market, we automatically
have a num\'eraire asset, so in this special case, the above generalisation
of the usual arbitrage theory
is not needed.  First recall some definitions:
\begin{definition}
A num\'eraire strategy is a pure-investment strategy $\eta$ such that
$$
 \eta_t \cdot P_t > 0 \mbox{ almost surely   for all } t \ge 0.
$$
\end{definition}

\begin{definition}
A risk-free strategy is a pure-investment strategy $\eta$ such that
$$
 \eta_t \cdot P_t  \mbox{ is $\ff_{t-1}$-measurable for all } t \ge 1.
$$
\end{definition}

\begin{theorem}   Suppose the market model is complete and is free of
investment-consumption arbitrage.  Then there
exists a risk-free num\'eraire strategy.
\end{theorem}

\begin{proof}
By completeness, zero-coupon bonds can be attained.  That is,
for each $T > 0$ there exists a pure-investment strategy $H^T$ such that $H^T_T  \cdot P_T = 1$ a.s.
By no-arbitrage, the bonds are num\'eraires: setting $B_t^T = H^T_t \cdot P_t$
we have $B_t^T > 0$ almost surely for all $0 \le t \le T$.  
Now define the money market account process $\beta$ by
$$
\beta_t = \prod_{s=1}^t (1 + r_s)
$$
where
$$
r_t = \frac{1}{B_{t-1}^{t}} - 1.
$$
Note that $\beta$ is predictable and strictly positive.  Furthermore, let $\eta_t = \beta_t H^t_t$.
This portfolio corresponds to holding the $\beta_t$ units of the bond with maturity $t$ during the 
period $(t-1,t]$ just before its maturity. 

First note that $\beta_t = \eta_t \cdot P_t$  since $B_t^t= H_t^t \cdot P_t = 1$.  
Finally note that the predictable process $\eta$ is a self-financing pure-investment strategy since
\begin{align*}
\eta_{t+1} \cdot P_t &=  \beta_{t+1} H^{t+1}_{t+1} \cdot P_t \\
 &=  \beta_{t+1} H^{t+1}_{t} \cdot P_t  \mbox{ (since $H^{t+1}$ is pure-invest.) }  \\
& = \beta_{t+1} B^{t+1}_t  \\
& = \beta_t
\end{align*}
as desired.
\end{proof}

\subsection{Num\'eraires and equivalent martingale measures}\label{se:num}
In this section, we recall the  concepts of a num\'eraire
  and an equivalent martingale measures.  The primary purpose of this section 
 is to reconcile concepts and terminology used by other authors.

We now recall the definition of an equivalent martingale measure.   
We will now see that the notions of martingale deflator and equivalent martingale measure  are essentially the 
same concept once a num\'eraire is specified. 

\begin{definition} 
Suppose there exists a num\'eraire
strategy   $\eta$ with corresponding wealth $ \eta \cdot P  = N$.  
An equivalent martingale measure relative 
to this num\'eraire is any probability measure 
$\QQ$ equivalent to $\PP$ such that the discounted asset prices $P/N$ 
 are  martingales under $\QQ$.
\end{definition}

\begin{proposition}
Let $Y$ be a martingale deflator for the model, and let 
$\eta$ be a num\'eraire strategy with value $\eta \cdot P  = N$.
Fix a time horizon $T > 0$, and define a new measure $\QQ$ by the density
$$
\frac{d\QQ}{d\PP} = \frac{N_T Y_T  }{N_0 Y_0}.
$$
Then $\QQ$ is an equivalent martingale measure  for the finite-horizon model $(P_t)_{0 \le t \le T}$.  

Conversely, let $\QQ$ be an equivalent martingale measure for $(P_t)_{0 \le t \le T}$.   Let
$$
Y_t = \frac{1}{N_t} \EE^{\PP} \left( \frac{d\QQ}{d\PP} | \ff_t \right).
$$
Then $Y$ is a martingale deflator.
\end{proposition}

\begin{proof}  First we need to show that
the proposed density does in fact define an equivalent probability measure.  Since $\eta$ 
is a pure-investment strategy and the $N$ is positive by definition, the process $NY$ is a martingale by 
Proposition \ref{th:mart}.  In particular, the random variable $d\QQ/d\PP$ above is positive-valued.  Also, since 
$NY$ is a martingale, we have
$$
\EE^{\PP} ( N_T Y_T  ) = Y_0 N_0  \ \ \Rightarrow \EE^{\PP} \left( \frac{d\QQ}{d\PP}  \right) = 1.
$$

Now we will show that the discounted price process $P/N$ is a martingale under $\QQ$.  For $0 \le t \le T$ 
 Bayes's formula  yields
\begin{align*}
\EE^{\QQ}\left( \frac{P_T}{N_T} | \ff_t \right) &= \frac{\EE^{\PP} \left( \frac{d \QQ}{d \PP} \frac{P_T}{N_T} | \ff_t \right)}
{\EE^{\PP} \left( \frac{d \QQ}{d \PP} | \ff_t \right)} \\
&= \frac{\EE^{\PP} \left( P_T N_T | \ff_t \right)}
{\EE^{\PP} \left( N_T Y_T  | \ff_t \right)} \\
& = \frac{ P_t}{N_t}
\end{align*}
since by the definition of martingale deflator,  both $PY$ and $NY$ are martingales.   

Now for the converse.   Let $\QQ$ be an equivalent martingale measure and $Y$ be defined by the formula.
\begin{align*}
P_t Y_t & = \EE^{\QQ}\left( \frac{P_T}{N_T} | \ff_t \right) \EE^{\PP} \left( \frac{d \QQ}{d \PP} | \ff_t \right) \\
& = \EE^{\PP} \left( Y_T P_T | \ff_t \right)  
\end{align*}
and hence $Y$ is a martingale deflator.
\end{proof}

\begin{corollary}  Consider a finite horizon market model with a num\'eraire.
There is  no pure investment arbitrage if and only if there exists an 
equivalent martingale measure
\end{corollary}

\begin{remark}
Note that an equivalent martingale measure as defined here only makes sense in the
context of a market model with some fixed, finite time-horizon.   In general, even
if there is no arbitrage there does not exist an equivalent measure under which
the discounted market prices are martingales over an infinite horizon, since 
the martingale $NY$ may fail to be uniformly integrable.  This 
technicality can be resolved by invoking the notion of a locally equivalent measure.

However, notice that the fundamental theorem of asset pricing, when 
stated in terms of the price density, holds for all time-horizons simultaneously. 
  \end{remark}

\section{Other notions of arbitrage and bubbles in discrete time}\label{se:abs-rel}

We now reconsider the definition of arbitrage
as defined above.  Indeed, there are a number of, a priori distinct, notions of arbitrage
which appear naturally in financial modelling.  

It is difficult to give a mathematically precise definition of a price bubble
in a financial market model.  One possible definition is to say that there
exists a bubble if there exists a weak notion of arbitrage but not a stronger notion.
  We now elaborate on this point.

Suppose we are given both a set of admissible trading strategies and  a
collection of preference relations on this set.  
An absolute arbitrage is a strategy $H^{\mathrm{abs}}$ such that for any admissible $H$, the strategy
 $H+H^{\mathrm{abs}}$ is also admissible and $H+H^{\mathrm{abs}}$ is strictly preferred to $H$
for all preference relations.  
An absolute  arbitrage is scalable in the sense that for all $k \ge 0$
the strategies $H^k = H+ k H^{\mathrm{abs}}$ are feasible and 
$H^{k+1}$ is preferred to strategy $H^k$.

It is usually considered
desirable to consider models without this type of arbitrage. Indeed, 
if prices are derived from a competitive
equilibrium, then all agents are holding their optimal allocation.
However,   if the market admits an absolute arbitrage, then 
there does not exist an optimal strategy for any agent: given any strategy, the 
agent can find another strategy that is strictly preferred.  

The notion of investment-consumption arbitrage as defined in section \ref{se:1ftap} is
that of absolute arbitrage.  In particular, notice that it
is the appropriate notion when our class of preference relations
are given by utility functions of the form given by equation \eqref{eq:U}
where the functions $c_t \mapsto u(c_0, \ldots, c_T)$ are strictly increasing.

To see how the choice of admissible strategies and preference relations affects this
notion of arbitrage, suppose that we consider 
investors who only receive utility from consumption at a
fixed date having utility functions of the form
$$
\tilde U(H) =  \EE[ u( H_T \cdot P_T ) ].
$$
An appropriate definition of arbitrage in this case is this:

\begin{definition} A terminal-consumption   arbitrage is an investment-consumption
arbitrage $H$ over the  time horizon $T > 0$ such that
$$
\PP\left(  H_T \cdot P_T > 0 \right) > 0.
$$
\end{definition}

In the above definition, we allow the investor to consume before the
terminal date; however, the investor does not receive any utility for
this early consumption.  That is, we have modified the set of preference
relations while fixing the given set 
$
\{ \mathcal{H}_{X_0,T}: X_0, T \}
$
of admissible strategies.  If we also modify the set of admissible
strategies by insisting that the investor does
not consume before the terminal date, we have yet another type of arbitrage:

\begin{definition} A pure-investment arbitrage is a   terminal-consumption   arbitrage
$H$ over the time horizon $T > 0$ such that $(H_t)_{0 \le t \le T}$ is a pure-investment
strategy.
\end{definition}

The following proposition shows that these various types of arbitrages
coincide when there exists a num\'eraire:

\begin{proposition}
Consider a marker for which there exists a num\'eraire strategy.
There exists an investment-consumption arbitrage if and only if
there exists a pure-investment arbitrage.
\end{proposition}

\begin{proof}
Let $\eta$ be a num\'eraire strategy with corresponding wealth process $\eta  \cdot P  = N$.  Let $H$ be a
self-financing investment-consumption strategy with $H_0=0$, and finally let $K$ be the strategy that
consists of holding at time $t$ the portfolio $H_t$ but of instead of consuming 
the amount $(H_t-H_{t+1})\cdot P_t$,
this money instead is invested into the num\'eraire portfolio.  In notation, $K$ is defined by
$$
K_t  = H_t + \eta_t   \sum_{s=1}^{t-1} \frac{ (H_s - H_{s+1})\cdot P_s }{N_s}  
$$
Note that 
\begin{align*}
(K_t - K_{t+1}) \cdot P_t =&   (H_t-H_{t+1}) \cdot P_t  -  \eta_{t+1} \cdot P_t \frac{  (H_t - H_{t+1})\cdot P_s }{N_{t}}  \\
& +  (\eta_t- \eta_{t+1}) \cdot P_t \sum_{s=1}^{t-1} \frac{  (H_s - H_{s+1})\cdot P_s}{N_s}\\
= & 0
\end{align*}
so $K$ is a pure investment strategy by the assumption that $\eta$ is pure-investment.  Finally 
that if $H_{T+1} = 0$, then
$$
K_T \cdot P_T =  N_T \sum_{s=1}^{T} \frac{ (H_s - H_{s+1})\cdot P_s}{N_s}  \ge 0.
$$
In particular, $K$ is a pure-investment arbitrage if and only if $H$ is an investment-consumption
arbitrage.
\end{proof}
 
Corresponding to these weakened notions of arbitrage are weakened notions
of martingale deflator. 
\begin{definition}
A signed martingale deflator is a (not necessarily positive) adapted process $Y$ such that $PY$ is
an $n$-dimensional martingale.
\end{definition}

A sufficient condition to rule out arbitrage can be formulated in this case.

\begin{theorem}\label{th:ftap-alt}   Suppose that for every $T > 0$, there exists a signed martingale deflator $Y^T = (Y_t^T)_{0 \le t \le T}$
such that $Y^T_T > 0$ almost surely.   Then there is no pure-investment arbitrage.   If
in addition, $Y^T_t \ge 0$ almost surely for all $0 \le t \le T$, then there is no terminal-consumption
arbitrage.
\end{theorem}

The proof makes use of Proposition \ref{th:mart}.  The details are omitted.

Just as we can consider martingale deflators as the dual variables
for investment-consumption utility maximisation problem
described in section \ref{se:invcon}, it is easy to see that we can consider signed
martingale deflators as the dual variables for the
 utility maximisation problem with the  pure-investment objective
$$
 \EE[ u( H_T \cdot P_T ) ]: \ H \mbox{ pure-investment with } H_0\cdot P_0 = X_0.
$$

Now that we have several notions of arbitrage available, we return
to the question of bubbles. Economically speaking,
a market has a bubble if there is an asset whose current price is higher
than some quantification of its fundamental value.  Of course, the fundamental value should
reflect in some way the future value of the asset.   Therefore, 
it is natural to say that a discrete-time market has a bubble if there
exists an investment-consumption arbitrage.   We will now give an
example of such a market that has the additional property that there is 
no terminal-consumption arbitrage and hence no pure-investment arbitrage. 
The idea is that  an agent who is obliged to be fully invested
in the market, such as the manager of a fund which is required to hold assets in a certain
sector, cannot take advantage of the `obvious' risk-less profit opportunity.  
In section \ref{se:relative} we will show that this situation is
analogous to the  continuous-time phenomenon when a market can
have no absolute arbitrage yet have a relative arbitrage, where again, the
obvious risk-less profit is impossible to lock in because of admissibility 
constraints.

Consider a market with one
asset where the price is given by $P_t=\one_{\{ t < \tau\} }$ for some positive, finite
stopping time $\tau$.  In some sense, the fundamental value of this asset is zero,
since $P_t = 0$ for all $t \ge \tau$.  The obvious strategy for an investor
to employ is to sell the asset short at time $0$, and consume the proceeds.  
Then at time $\tau$, the investor buys the asset back from the market at no cost.

We consider two cases.  First, if $\tau$ is unbounded, this strategy is not an investment-consumption
arbitrage according to our definition, since we require an arbitrage to be
concluded at a non-random time $T$.  Indeed, suppose that 
 $\tau$ is not only unbounded but also that on the
event $\{t-1 < \tau \}$ the conditional probability $\PP( t < \tau | \ff_{t-1} )$ is
strictly positive almost surely for all $t \ge 1$.  In this case we can find a martingale deflator $Y$ by
defining
$$
Y_t = \prod_{s=1}^{t \wedge \tau} \PP(s < \tau | \ff_{s-1} )^{-1}.
$$ 

Now suppose that $\tau$ is bounded by a constant $N$, so that $\tau \le N$ almost surely.
There is an investment-consumption arbitrage:  simply 
sell short one share of the asset and consume the proceeds. 
In notation, let $H_t= -1$ for $0 \le t \le N$ and $H_{N+1} = 0$.
The corresponding consumption strategy is $C_0 = 1$ and $C_t = 0$ for $1 \le t \le N$.

On the other hand,  since
there is no num\'eraire asset, there is no way to lock in this arbitrage with a terminal-consumption
strategy.  Indeed, if $H_0 P_0 = 0$ then $H_t P_t \le 0$ for all $t \ge 0$.
 
 One might interpret this example as a discrete-time market with a bubble. 
There might be economic grounds for the existence of such bubbles if sufficiently
many traders can not withdraw gains from trade from a market
account before some specified time horizon.

Alternatively we can see that 
there is no terminal-consumption arbitrage by using Theorem \ref{th:ftap-alt}, by finding a family of 
non-negative signed martingale deflators,
 even when $\tau$ is almost surely bounded by a non-random time $N$. 
Indeed, for the case where $T < N$, let $Y^T_t = 1$ for all $0 \le t \le T$,
and for the case where $T \ge N$, let  $Y_t^T = \one_{\{ 0 \le t < T \}}$.
 Note that
$P Y^T   $ is a martingale  in both cases since
$P_t Y_t^T = 1$ in the $T< N$ case and $P_t Y_t^T = 0$ in the $T \ge N$
case.

\section{Relative arbitrage and bubbles in continuous time}\label{se:relative}
In this section we discuss the notion of relative arbitrage and a 
popular definition of bubble in a continuous-time market.  The results here
are not new, but are included to provide context to the discussion in the previous
section.

In contrast to the notion of an absolute arbitrage,
 a relative arbitrage can be described as follows.  As before,
 we are given both a set of admissible trading strategies and  a
collection of preference relations on this set.
 An arbitrage relative to the benchmark admissible strategy $H$ is 
a strategy $H^{\mathrm{rel}}$ such that $H+H^{\mathrm{rel}}$ is admissible and preferred to $H$ for all preference relations.
Note that unlike the case of absolute arbitrage, a relative arbitrage is not
necessarily scalable since there is no guarantee that $H+kH^{\mathrm{rel}}$ is admissible 
for $k > 1$.  
As in the case of absolute arbitrage, it might also 
be desirable to exclude relative arbitrage from a model, but the argument is weaker.
For instance, if there is a relative arbitrage, then in equilibrium, no
agent would implement the benchmark strategy $H$.  In particular, if $H$ is a buy-and-hold
strategy for one of the assets, then no agent in equilibrium would hold a static position in that asset.

For our discrete-time models, it would be natural to say that 
$H^{\mathrm{rel}}$ is an arbitrage relative to $H$ if 
the initial cost $H^{\mathrm{rel}} \cdot P_0 = 0$ vanishes and the 
consumption stream associated to $H+H^{\mathrm{rel}}$ dominates
the consumption stream associated to $H$, with strict domination
with strictly positive probability.  It easy to see that in
this case $H^{\mathrm{rel}}$ is also an absolute arbitrage.

Therefore, we turn our attention briefly to  continuous-time models.
We consider a market with a continuous semimartingale price process $P$.

First we define the set of self-financing investment-consumption strategies
$$
\mathcal{A}^{*} = \left\{ H:  P\mbox{-integrable}, H\cdot P -  \int H \cdot dP \mbox{ is decreasing } \right\}.
$$
To avoid extraneous complication and highlight the difference between 
relative and absolute arbitrage, we now assume that there exists a num\'eraire strategy.
Recall that in the discrete-time setting this implies that the various notions
of arbitrage discussed in the last section coincide.  In particular, we 
consider only the case of pure-investment arbitrage for simplicity.
The appropriate set of self-financing pure-investement strategies becomes
$$
\mathcal{A}^{\circ} = 
\left\{ H:  P\mbox{-integrable}, H_t\cdot P_t= H_0\cdot P_0 +  \int_0^t H_s \cdot dP_s \mbox{ a.s. for all } t \ge 0 \right\}.
$$

Unlike the discrete-time setting, it is well known that in continuous time we must restrict the
strategies available to investors in order to avoid trivial arbitrages
arising from doubling strategies.  Therefore, we 
assume that for every admissible strategy, the investor's wealth remains
non-negative.   That is, we let
$$
\mathcal{A} = \{ H \in \mathcal{A}^\circ:  H_t\cdot P_t \ge 0 \mbox{ a.s. for all } t \ge 0 \}.
$$
 If $H$ is a given admissible strategy, then a  relative
arbitrage $H^{\mathrm{rel}}$ has the property that the wealth generated by
$H+ H^{\mathrm{rel}}$ is non-negative at all times; that is,
we have
$$
H^{\mathrm{rel}}_0 \cdot P_0 = 0 \mbox{ and } H^{\mathrm{rel}}_t \cdot P_t \ge - H_t \cdot P_t \mbox{ a.s. for all } t \ge 0
$$
On the other hand, a candidate absolute arbitrage  $H^{\mathrm{abs}}$ should be an arbitrage relative to any admissible
$H$, and hence the wealth it generates should be non-negative:
$$
H^{\mathrm{abs}}_0 \cdot P_0 = 0 \mbox{ and } H^{\mathrm{abs}}_t \cdot P_t \ge 0 \mbox{ a.s. for all } t \ge 0
$$

We state here a sufficient condition to rule out arbitrage.

\begin{proposition} \label{th:no-relarb}  There is
no absolute arbitrage if there exists a positive continuous semimartingale $Y$ such that
$YP$ is a  local martingale.  There is
no arbitrage relative to an admissible strategy $H$ if the process $H\cdot P Y$ is a
true martingale.  
\end{proposition}

\begin{proof}
Let $K$ be an admissible strategy. 
By the It\^o's formula, the Kunita--Watanabe formula and the self-financing condition
we have
$$
K_t \cdot P_t Y_t = K_0 \cdot P_0 Y_0 + \int_0^t K_s \cdot d(P_s Y_s).
$$
In particular, by the integral representation on the right-hand side, we have that
$K \cdot P Y$ is a local martingale.  And by admissibility, this local martingale
is non-negative hence is a supermartingale by Fatou's lemma.  In particular,
$$
\EE (K_T \cdot P_T Y_T ) \le  K_0 \cdot P_0 Y_0. 
$$
Now letting $K = H + H^*$ where $H \cdot P Y$ is a true martingale and $H^*_0 \cdot P_0 = 0$,
we have
$$
\EE (H^*_T \cdot P_T Y_T ) \le  0. 
$$
Hence there is no arbitrage relative to $H$.  Since we may let $H=0$, there is no absolute
arbitrage.
\end{proof}

\begin{remark}
The notion of an absolute arbitrage used here is closely related to the   
 num\'eraire-independent property of no arbitrage
of the first kind (NA1).  Recently, a converse
to the above proposition has been proven, that the market model has NA1
if and only if there exists a local martingale deflator, 
in the one-dimensional case by Kardaras \cite{K} and in the multi-dimensional
case by Schweizer \& Takaoka  \cite{SchweizerTakaoka}.
\end{remark}

\begin{remark}  Note that the process $M = H\cdot P Y$ is always a local martingale
when $Y$ is a local martingale deflator.  It is a true martingale if and only if $M$
is of class DL, that is, the collection of random variables
$$
\{ M_{\tau \wedge t}: \tau \mbox{ a stopping time }\}
$$
is uniformly integrable for all $t \ge 0$.   When $M$ is a true martingale and $H$ is
a num\'eraire strategy,  one can define an equivalent martingale measure relative to 
$H$ as described in section \ref{se:num}.

The notion of a relative arbitrage is closely related to no free lunch with vanishing risk (NFLVR).
Indeed, consider the case when the reference strategy $H$ is a num\'eraire, 
so that it generates a strictly positive wealth process $N$. In this case,  
 a candidate relative arbitrage $H^*$ is such that the discounted wealth $H^* \cdot P/N$ is bounded
from below by the constant $-1$.   In a celebrated paper of
 Delbaen \& Schachermayer \cite{DelSch} proved another converse of the above proposition,
that a market model has  NFLVR 
if and only if there exists an equivalent sigma-martingale measure. 
\end{remark}

  A typical example of a market with a relative arbitrage but no absolute arbitrage
	has  two
assets.  The first is cash with constant unit
price, and the second is a risky stock with positive price process $S$. 
 Suppose $S$ is a local
martingale.  

The process $P=(1,S)$ is a local martingale, so we can take  $Y=1$ to be a local martingale deflator. 
By Proposition \ref{th:no-relarb}  there cannot be an absolute
arbitrage.  Furthermore, since the value of holding a static position of cash is constant, 
there can be no arbitrage relative to the strategy $(1,0)$. 

However, suppose now that $S$ is a strictly local martingale (and hence a supermartingale),
that the filtration is generated by a Brownian motion and that the volatility of $S$ is strictly positive.
Then there does exist a relative arbitrage relative to the strategy $(0,1)$ of holding one share of the stock.  Indeed, for any fixed
horizon $T > 0$ there exists a pure-investment trading strategy such that
$$
H^{\mathrm{rep}}_t\cdot P_t = \EE(S_T | \ff_t) \le S_t
$$
by the martingale representation theorem.  Note that the strategy $H^* = H^{\mathrm{rep}}-(0,1)$,
that is longing the dynamic replication strategy and shorting the stock,
is a relative arbitrage.  It is not an absolute arbitrage since
$H^*$ is itself not admissible. 

 The phenomenon exhibited by this example has been proposed to model
price bubbles, since the simple strategy of buying and holding the stock
is dominated by a dynamic replication strategy $H^{\mathrm{rep}}$. 
 See the recent paper of Herdegen \cite{Herdegen} or the presentation of Schweizer \cite{schweizer} for a 
discussion of this point.
However, the above example uses in a fundamental way the special properties of continuous
time.  Indeed, if $S$ is a positive local martingale in \textit{discrete} time,
then $S$ is automatically a true martingale by Proposition \ref{th:non-neg}.  

 \begin{remark}  There is a tantalising parallel between the continuous-time
bubbles discussed above and the discrete-time bubbles of the last section.   Indeed,
in continuous time, bubbles arise when the positive process $M =  NY$ is a strictly
local martingale, where $N$ is a num\'eraire and $Y$ a local martingale deflator.
Recall that a continuous positive strictly local martingale $M$ can be constructed as follows.

Let $X$ be a continuous non-negative true martingale with respect to a measure $\PP$, where $X_0 =1$.
Fix a time horizon $T > 0$ and define an absolutely continuous measure $\QQ$ with
density 
$$
\frac{d \QQ}{d\PP} = X_T.
$$
Now let $\tau = \inf\{ t \ge 0: X_t = 0 \}$ be the first time that $X$ hits zero.
Finally, let $M$ be defined as
$$
M_t = \frac{ \one_{\{ t < \tau\}}}{X_t}.
$$
The process $M$ is a $\QQ$-local martingale.  Indeed,  it is easy to check that
the sequence of stopping times $\tau_n = \inf\{ t \ge 0: X_t = 1/n \}$ localises $M$ to a bounded
$\QQ$-martingale.  Furthermore, $M$ is strictly positive $\QQ$-almost surely since 
\begin{align*}
\QQ( \tau \le T ) & = \EE^{\PP} (X_T \one_{ \{ \tau \le T \}} ) \\
&= \EE^{\PP} (X_\tau \one_{ \{ \tau \le T \}} ) \\
&= 0.
\end{align*}
However, note that 
$$
\EE^{\QQ}( X_T ) = \PP( \tau > T ).
$$
In particular, $M$ is a true $\QQ$-martingale if and only if $X$ is 
strictly positive $\PP$-almost surely. See, for instance, the paper of Ruf \& Rungaldier  \cite{RR} for further details.

As a consequence of the above discussion,
in the continuous-time story, there is a bubble if a 
certain process $X$ hits zero with positive probability.
On the other hand, in the discrete-time theory  there is 
no terminal consumption arbitrage if there is a non-negative signed
martingale deflator $Y$.  However, there may
exist an investment-consumption arbitrage--that is, a bubble-- if $Y$ hits zero.
\end{remark}
\section{The proofs}\label{se:proofs}

\subsection{Proof of Theorem \ref{th:1FTAP}}

Recall that we need only prove that (1) implies both (2) and (4).  
First, we need an additional equivalent, but more technical, formulation of Theorem \ref{th:1FTAP}:

\begin{theorem}\label{th:1FTAP-tech}  Conditions (1)-(4) of Theorem \ref{th:1FTAP} is equivalent to
\begin{enumerate}
\item[(5)] For every $t \ge 1$ and positive $\ff_t$-measurable $\zeta$, there exists a positive
 $\ff_t$-measurable random variable $Z$ and positive $\ff_{t-1}$-measurable random variable $R$
 such that
 $$
 Z  \le R \zeta  \mbox{ almost surely}
 $$
 and 
 $$
 \EE( P_t Z  | \ff_{t-1}) = P_{t-1}
$$
 in the sense of generalised conditional expectation.
\end{enumerate}
\end{theorem}

\begin{proof}  [Proof of (5) $\Rightarrow$ (2) of Theorem \ref{th:1FTAP-tech}]
For each $t \ge 1$, let $Z_t$ be such that 
$$
 \EE( P_t Z_t  | \ff_{t-1}) = P_{t-1}
$$
 in the sense of generalised conditional expectation.  Let $Y_0 = 1$ and
$Y_t =  Z_1 \cdots Z_t$.   Note that by Proposition \ref{th:locmart} stated
in section \ref{se:tech} below, the process $YP$ is a local martingale.   
Hence $Y$ is a local martingale deflator.
\end{proof}

\begin{proof}  [Proof of (5) $\Rightarrow$ (4) of Theorem \ref{th:1FTAP-tech}]
By replacing, the given positive process $\eta$ with $\hat \eta$ defined by
$$
\hat \eta_t = \min\{ \eta_t, e^{-\| P_t\| } \}.
$$
we can assume that the process $P\eta$ is bounded.   In particular,
once we show that given  $\eta = (\eta_t)_{0 \le t \le T}$ there
exists a local martingale deflator such that $Y_t \le \eta_t$ for $0 \le t \le T$, 
we can conclude that $Y$ is a true martingale deflator since the process
 $PY$ is integrable.

Given  the process $\eta = (\eta_t)_{0 \le t \le T}$ we will construct 
random variables $(Z_t)_{1 \le t \le T}$ such that 
the process $(Y_t)_{0 \le t \le T}$ is a local martingale deflator, where $Y_t = Y_0 Z_1 \cdots Z_t$. 
We need only show that we can do this construction in such a way that $Y_t \le \eta_t$.

Let $\zeta_T = \eta_T/\eta_{T-1}$.  By condition (5), there exists a positive random
variable $Z_T$ and a positive $\ff_{T-1}$-measurable random variable $R_{T-1}$
such that
$$
Z_T \le R_{T-1} \zeta_T.
$$
Now we proceed backwards by specifying $\zeta_{T-1},  \ldots, \zeta_1$ 
to find $Z_{T-1}, \ldots, Z_1$ and corresponding bounds $R_{T-2}, R_{T-3}, \ldots, R_0$ such that
$$
Z_t \le R_{t-1} \zeta_t
$$
by letting
$$
\zeta_t = \frac{\eta_t}{\eta_{t-1}} \left( \frac{1}{1+ R_t} \right).
$$
The process $Y_t =  \frac{\eta_0}{1+ R_0} Z_1 \cdots Z_t$ is a local martingale deflator
such that $Y_t \le \eta_t$ for $0 \le t \le T$, as desired. 
\end{proof}

We now come the converse direction.  The following proof is adapted from Rogers's
proof \cite{rogers} of the Dalang--Morton--Willinger theorem in the case 
when a num\'eraire asset is assumed to exist.  The idea there is to show that no arbitrage implies a certain
pure-investment utility optimisation problem has an optimal solution.  The idea here is very similar:
no arbitrage
implies that a certain investment-consumption utility maximisation problem has an optimal solution.

\begin{proof} [Proof of (1) $\Rightarrow$ (5) of Theorem \ref{th:1FTAP-tech}]
Fix $t \ge 1$, and suppose the positive $\ff_t$-measurable  random variable $\zeta$ is
given.  By replacing $\zeta$ with $\zeta \wedge 1$ there is no loss assuming 
$\zeta$ is bounded.  Let
$$
\hat \zeta =  \zeta e^{- \| P_t\|^2/2}
$$ 
 and define a function $F: \RR^n \times \Omega \to \RR$
by
$$
F(h) = e^{ h \cdot P_{t-1} } + \EE[ e^{- h \cdot P_t } \hat \zeta | \ff_{t-1}].
$$
More precisely, let $\mu$ be the regular conditional joint distribution of $(P_t, \zeta)$ given $\ff_{t-1}$, and  let
$$
F(h, \omega) = e^{h \cdot P_{t-1}(\omega)} + \int e^{-h \cdot y - \|y\|^2/2} z  \mu( dy, dz, \omega) .
$$
Note that $F(\cdot, \omega)$ is everywhere finite-valued, and hence is smooth.
We will show that no investment-consumption arbitrage implies that for
each $\omega$, the function $F(\cdot, \omega)$ has a minimiser $H^*(\omega)$
such that $H^*$ is  $\ff_{t-1}$-measurable.   By the first order condition for a minimum, we have
$$
0 = \nabla F( H^*) =  e^{ H^* \cdot P_{t-1} } P_{t-1} - \EE[ e^{- H^* \cdot P_t } \hat{\zeta} P_t | \ff_{t-1}]
$$
and hence we may take
$$
Z = e^{ - H^* \cdot P_{t-1} - H^* \cdot P_t} \hat{\zeta}.
$$
Note that 
$$
Z \le R \zeta,
$$
where $R$ is the $\ff_{t-1}$-measurable random variable
$$
R =e^{ - H^* \cdot P_{t-1} + \| H^* \|^2/2 }.
$$

With the above goal in mind, we define  functions $F_k: \RR^n \times \Omega \to \RR$ by
$$
F_k(h, \omega) = F(h, \omega) +  \|h \|^2/k.
$$
Now for fixed $\omega$, the function $F_k(\cdot, \omega)$ is smooth, strictly convex and 
$$
F_k(h, \omega) \to \infty \mbox{ as } \|h \| \to \infty.
$$
In particular, there exists a unique minimiser $H_k(\omega)$, and by Proposition \ref{th:meas} 
$H_k$ is $\ff_{t-1}$-measurable.    

We will make use of two observations.  
First, note that $H_k$ enjoys a certain
non-degeneracy property.  To describe it, let
$$
\mathcal{U}(\omega) = \{ u \in \RR^n:  u \cdot P_{t-1}(\omega) = 0, 
\PP( u\cdot P_t = 0 | \ff_{t-1})(\omega) = 1 \} 
$$
and let $\mathcal{V}(\omega) = \mathcal{U}(\omega)^{\perp}$.  Note
that the minimiser $H_k(\omega)$ is in $\mathcal{V}(\omega)$ for each $\omega$ since
$$
F(u+v) = F(v)
$$
and hence
$$
F_k(u + v) \ge F_k(v)
$$
whenever $u \in \mathcal{U}$ and $v \in \mathcal{V}$.  

Second, note that  $F(H_k) \to \inf_h F(h)$ almost surely, since
\begin{align*}
\limsup_k F(H_k) & \le \limsup_k F_k(H_k) \\
& \le \limsup_k F_k(h) \\
& = F(h) 
\end{align*}
for all $h \in \RR^n$.  

Now, let 
$$
A =  \{ \sup_k \|H_k\| < \infty \}
$$
be the $\ff_{t-1}$-measurable set on which the sequence $(H_k(\omega))_k$ is bounded.   Hence,
by Proposition \ref{th:select} we can extract a measurable subsequence on which that $H_k$ converges on $A$
to a $\ff_{t-1}$ measurable $H^*$.  Note that by the smoothness of $F$, we have
\begin{align*}
F(H^*) &= \lim_k F(H_k) \mbox{ on } A
\end{align*}
 Hence $H^*$ is a minimiser of $F$ by the second observation, 
and the proof is complete once we show that $\PP(A) = 1$.  We 
will make use of the fact that $H^* \in \mathcal{V}$ by the first observation.

We now show $\PP(A^c) = 0$.   
Now on $A^c$ the sequence $( H_k)_k$ is unbounded.  Hence, we can find a measurable subsequence
along which $\|H_k \| \to \infty$ on $A^c$.  Since the sequence
$$
\hat H_k = \frac{ H_k }{\| H_k \| }
$$ 
is bounded, and indeed $\| \hat H_k \| =  1$ for all $k$, there exists a further subsequence along which
$(\hat H_k)_k$ converges on $A^c$ to a $\ff_{t-1}$-measurable random variable $\hat H$.  Note that
$\| \hat H(\omega) \| = 1$ for all $\omega \in A^c$.    Letting
\begin{equation}\label{eq:H*}
\bar H = \hat H \one_{A^c},
\end{equation}
we need only show that $\bar H = 0$ almost surely.

First note that on $A^c \cap  \{   \hat H  \cdot P_{t-1} < 0 \}$ we have
$$
e^{  H_k \cdot P_{t-1} } = (e^{ -\hat H_k \cdot P_{t-1} })^{\|H^*_k\|} \to \infty.
$$
But since 
\begin{align*}
\limsup_k e^{  H_k \cdot P_{t-1} } &\le \limsup_k F(H_k) \\
& = \inf_h F(h) \\
& \le F(0) = 2
\end{align*}
by the second observation above, we conclude that 
\begin{equation}\label{eq:t-1}
\PP(A^c \cap  \{   \hat H  \cdot P_{t-1} > 0 \}) = 0.
\end{equation}

Similarly, on $A^c$ we have by Fatou's lemma and Markov's inequality  
\begin{align*}
\PP(  \hat H  \cdot P_t < 0 |\ff_{t-1} ) &= \sup_{\varepsilon > 0 } \PP(  \hat H  \cdot P_t < - \varepsilon |\ff_{t-1} ) \\
& \le \sup_{\varepsilon > 0 } \liminf_k \PP(  -H_k  \cdot P_t  > \|H_k\| \varepsilon  | \ff_{t-1} ) \\
& \le \sup_{\varepsilon > 0 } 
\liminf_k \frac{ \EE( e^{ -H_k  \cdot P_t } \hat \zeta | \ff_{t-1} )}{  e^{ \|H_k\| \varepsilon} \EE( \hat \zeta | \ff_{t-1}) } \\
& = 0
\end{align*}
since 
$$
\limsup_k  \EE( e^{ -H_k  \cdot P_t } \hat \zeta | \ff_{t-1} )  \le \limsup_k F(H_k) \le 2.
$$
In particular, we can conclude that 
\begin{equation}\label{eq:t}
\PP( A^c \cap \{\hat H  \cdot P_t < 0 \}) = 0.
\end{equation}

Now consider the investment-consumption strategy $(H_s)_{0 \le s \le t}$ defined by
$H_s = 0$ for $0 \le s \le t-1$ and $H_t = \bar H $ where $\bar H$ 
is defined by equation \eqref{eq:H*}.
Note that by
equations \eqref{eq:t-1} and \eqref{eq:t} this strategy is self-financing.
By assumption that there is no investment-consumption arbitrage, we now conclude that
\begin{align*}
1 &=  \PP( \bar H\cdot P_{t-1} = 0, \ \bar H  \cdot P_t = 0 ) \\
& = \PP( \bar H\cdot P_{t-1} = 0, \  \PP( \bar H  \cdot P_t = 0 | \ff_{t-1} ) = 1 ) \\
& \le \PP( \bar H = 0) = \PP(A)
\end{align*}
where we have used the observation that $\bar H \in \mathcal V$
plus the fact that $\| \bar H \| = \one_{A^c}$ between the second
and third line.   This concludes the proof.
\end{proof}

\subsection{Proof of Theorem \ref{th:superrep}}
The dual characterisation of super-replicable claims can be proven
using the same utility maximisation idea as in the proof of Theorem \ref{th:1FTAP-tech}
above.  Indeed, the insight is that super-replication is the optimal
hedging policy for an investor in the limit of large risk-aversion.

\begin{proof}[Proof of Theorem \ref{th:superrep}]
Fix $t \ge 1$, and suppose that
$$
\EE( Z \xi_t | \ff_{t-1} ) \le \xi_{t-1}
$$
for any $\ff_t$-measurable $Z$ such that 
$$
\EE( P_t Z | \ff_{t-1} )  = P_{t-1}.
$$
We will show that there exists an $\ff_{t-1}$-measurable $H$ such that
$$
H \cdot P_{t-1} \le \xi_{t-1} \mbox{ and } H \cdot P_t \ge \xi_t \mbox{ almost surely}.
$$
For the sake of integrability, we introduce a factor $\zeta = e^{-(\|P\|_t^2 + \xi_t^2)/2}$
and define a family of random functions
$$
F_\gamma(h) = e^{- \gamma (\xi_{t-1} -h\cdot P_{t-1}) } + \EE[ e^{ -\gamma ( h \cdot P_t - \xi_t) } \zeta| \ff_{t-1}],
$$
where $\gamma > 0$ has the role of risk-aversion parameter.  
Since there is no arbitrage, we can reuse the argument from the proof of the
(1) $\Rightarrow$ (5) implication of Theorem \ref{th:1FTAP-tech}
to conclude that $F_\gamma$ has a $\ff_{t-1}$-measurable minimiser $H_\gamma$ with 
a corresponding $\ff_t$-measurable random variable
$$
Z_\gamma = \frac{ e^{ -\gamma (H_\gamma \cdot P_1 - \xi_t) } \zeta}{e^{- \gamma (\xi_{t-1} -H_\gamma\cdot P_{t-1}) }}
$$
with the property that 
$$
\EE( P_t Z_\gamma | \ff_{t-1} )  = P_{t-1}.
$$
Note that 
\begin{align*}
\frac{\partial}{\partial \gamma} F_\gamma( h) |_{h=H_\gamma} 
& = {e^{- \gamma (\xi_{t-1} -H_\gamma\cdot P_{t-1}) }}\big[( \EE( \xi_t  Z_\gamma | \ff_{t-1}) - \xi_{t-1})
+ H_\gamma \cdot ( P_{t-1} - \EE(P_t  Z_\gamma | \ff_{t-1}) \big] \\
& \le 0
\end{align*} 
by assumption. 

Also note that $\gamma \mapsto H_\gamma$ is differentiable.  Indeed, recall
that $H_\gamma$ is defined as the root of the function $\nabla F_\gamma: \mathcal V \to \mathcal V$, 
and $D^2 F_\gamma$ is a strictly positive definite operator on $\mathcal{V}$, so the
differentiability of $H_\gamma$ follows from the implicit function theorem. 
Furthermore, 
$$
F_{\gamma}(H_{\gamma}) \le F_{\gamma}(H_{\gamma \pm \varepsilon})
$$
since $H_{\gamma}$ is the minimiser of $F_{\gamma}$
and hence
$$
\frac{\partial}{\partial \gamma} F_g( H_\gamma) |_{g= \gamma} = 0.
$$
Putting this together implies $\gamma \mapsto F_\gamma(H_\gamma)$ is nonincreasing, and in particular
$$
\sup_{\gamma \ge 1} F_\gamma(H_\gamma) < \infty.
$$
where  for the rest of the proof we will only consider positive integer values for $\gamma$. 
Let 
$$
A = \{ \sup_{\gamma} \| H_\gamma\| < \infty \}
$$
be the $\ff_{t-1}$ event on which the sequence $( H_\gamma)_\gamma$ is bounded.  By
Proposition \ref{th:select} there exists a measurable subsequence which converges to 
$\ff_{t-1}$-measurable $H^*$ on $A$.   

Note that 
$$
\PP\left( A \cap \left[ \{ H^*\cdot P_{t-1} > \xi_{t-1} \} \cup \{H^*\cdot P_t < \xi_t \} \right] \right)  
\le \PP( F_\gamma(H_\gamma) \to \infty) = 0
$$
We need only show that $\PP(A)=1$.

We now prove that the event $A^c$ on which the sequence $(H_{\gamma})_{\gamma}$ is unbounded
has probability zero.  This follows the same steps as in the proof of Theorem \ref{th:1FTAP-tech}.
Again by Proposition \ref{th:select} we 
consider a subsequence such that $\| H_{\gamma} \| \to \infty$ and then
find a further subsequence and $\ff_{t-1}$-measurable random variable with $\| H_\gamma\|=1$ such that
$$
\hat H_\gamma \to \hat H
$$
where 
$$
\hat H_\gamma = \frac{H_{\gamma}}{\| H_{\gamma} \|}.
$$
Since 
$$
F_\gamma(H_\gamma) = (e^{ \hat H_\gamma \cdot P_{t-1} - \hat{\xi}_{t-1}})^{\gamma \|H_\gamma\|}
+ \EE[ (e^{  \hat{\xi}_t  - \hat H_\gamma \cdot P_t} )^{\gamma \|H_\gamma\|} \zeta| \ff_{t-1}],
$$
where 
$$
\hat \xi_s  = \frac{\xi_s}{\| H_{\gamma} \|} \to 0 \mbox{ for } s=t-1, t,
$$
we can conclude by the boundedness of $F_\gamma(H_\gamma)$ that
$$
\PP\left( A^c \cap \left[ \{ \hat H \cdot P_{t-1} > 0 \} \cup \{ \hat H \cdot P_t < 0 \} \right] \right)  = 0.
$$
As before, we use the assumption of no arbitrage and the non-degeneracy of 
$\hat H$ to conclude that $\PP(A^c) = 0$.
\end{proof}

\section{Appendix}\label{se:tech}
\subsection{Generalised conditional expectations and local martingales}
In this appendix we recall some basic notions and useful facts regarding discrete-time local martingales.

First we briefly recall a definition of the conditional expectation:

\begin{definition} Given a probability space $(\Omega, \ff, \PP)$, 
the non-negative random variable $X$ and sub-sigma-field $\mathcal{G} \subseteq \ff$, we define
the conditional expectation $\EE(X | \mathcal{G})$ to be the almost surely unique 
$\mathcal G$ measurable random variable $Y$ such that 
$$
\EE(X \one_G ) = \EE( Y \one_G)
$$
for all events $G \in \mathcal G$.  
If $\EE( |X| \ | \mathcal{G} ) < \infty$ almost surely, we define
$$
\EE(X | \mathcal{G}) =  \EE(X^+ | \mathcal{G})-\EE(X^- | \mathcal{G}).
$$
\end{definition}

Note that it is not necessary for $X$ to be integrable in order to define
the conditional expectation $\EE(X | \mathcal{G})$ as above.  However,
one must take care with this generalised notion of conditional expectation since
some of the familiar rules of calculation for integrable random variables
fail in general.
For instance, it  is possible to find a random variable $X$ and sigma-fields $\mathcal{G}$
and $\mathcal{H}$ such that both conditional expectations
$\EE(X | \mathcal{G})$ and $\EE(X | \mathcal{H})$ are defined,
and in fact, both are integrable, and yet
$$
\EE[\EE(X | \mathcal{G})] \ne \EE[\EE(X | \mathcal{H})].
$$
Of course, such pathologies do not occur if $X$ is integrable.

Below are three  useful properties of the generalised conditional expectation.
They are probably familiar to the reader, but are included here
to avoid unexpected pathologies as described above.
Their proofs can be found in Chapter 4.2 of \c{C}inlar's book \cite{cinlar}.

\begin{proposition}[Tower property]\label{th:tower}
Fix nested sigma-fields $\mathcal{G} \subseteq \mathcal{H}$.
Suppose $X$ is such
 either $X$ is almost surely non-negative or $\EE(|X| \ | \mathcal{G} )$ is almost surely finite,
then  $\EE( X | \mathcal{H} ) $ is defined and
$$
\EE(\EE( X | \mathcal{G} ) | \mathcal{H}) = \EE(\EE( X | \mathcal{H} ) | \mathcal{G}) 
=  \EE( X | \mathcal{G} ). 
$$
\end{proposition}

\begin{proposition}[Slot property] \label{th:slot}
Suppose $K$ is $\mathcal{G}$-measurable.  If either $\EE(|X| \ |\mathcal G) < \infty$
almost surely, or $X \ge 0$ and $K \ge 0$ almost surely,  then 
$$
\EE(K X | \mathcal{G})  = K \EE(  X | \mathcal{G}).
$$
\end{proposition}

\begin{proposition}  
Let $\mu$ be the regular conditional distribution of random variable
$X$ given a sigma-field $\mathcal{G}$.  Assuming the conditional
expectation is defined, we have
$$
\EE(X | \mathcal{G})(\omega) = \int x \ \mu( \omega, dx)
$$
for almost all $\omega \in \Omega$.
\end{proposition}

Now recall the definition of a local martingale:
\begin{definition} A local martingale  is an adapted process $M = (M_t)_{t \ge 0}$,
 such that there exists an increasing
 sequence of stopping times $(\tau_N)$ with $\tau_N \uparrow \infty$ such that
the stopped process $M^{\tau_N} =(M_{t \wedge \tau})_{t \ge 0}$ is a martingale for each $N$.
\end{definition}

The following  three properties of discrete-time local martingales  will be useful.
Their proofs  can be found, for instance, in the paper by Jacod \& Shiryaev \cite{JS}.
\begin{proposition}\label{th:locmart} A discrete-time process $M$ is a local martingale if and only
if  for all $t \ge 0$ we have $\EE( |M_{t+1}| \ | \ff_t ) < \infty$ almost surely and 
$\EE( M_{t+1} | \ff_t ) = M_t$ in the sense of generalised conditional expectation. 
\end{proposition}

\begin{proposition}\label{th:mart-trans} Suppose $Q$ is a discrete-time local martingale and $K$ is a predictable process.
Let  
$$
M_t =   \sum_{s=1}^t K_s (Q_s-Q_{s-1})
$$
for $t \ge 1$.  Then $M$ is a local martingale.   
\end{proposition}

\begin{proposition}\label{th:non-neg}  Suppose $M$ is a discrete-time local martingale such that
there is a non-random time horizon $T > 0$ with the property that   $M_T \ge 0$ almost surely. 
Then $(M_t)_{0 \le t \le T}$ is a martingale.
\end{proposition}

We conclude with a useful technical result due to Kabanov \cite{Kabanov}.  
Unlike the other results in this note, the proof of Kabanov's theorem requires subtle ideas
from functional analysis.  Fortunately, it is only used in one place - to show that
the existence of a local martingale deflator implies the existence of a true martingale deflator.
It should be stressed that Kabanov's theorem is only needed in the proof 
because we do not restrict our attention 
to a finite time horizon and because we do not assume that our assets have non-negative prices.

\begin{theorem}[Kabanov \cite{Kabanov}]\label{th:kabanov} Suppose $M$ is a discrete-time local martingale with respect to a
probability measure $\PP$.  Then there exists an equivalent measure $\QQ$, such that $M$
is a true martingale with respect to $\QQ$.
\end{theorem}

\subsection{Measurability and selection}
The following result allows us to assert the measurability of the minimiser of a random function.
The proof is from  the paper of Rogers \cite{rogers}.
\begin{proposition}\label{th:meas}  Let  $f: \RR^n \times \Omega \to \RR$ is such that $f(x, \cdot)$ is measurable for all $x$, and
that $f(\cdot, \omega)$ continuous and has a unique minimiser $X^*(\omega)$ for each $\omega$.  Then $X^*$ is
measurable.
\end{proposition}

\begin{proof}
For any open ball $B \subset \RR^n$ we have
$$
\{\omega:  X^*(\omega) \in B \} = \bigcup_{p \in B \cap Q} \bigcap_{q \in B^c \cap Q}\{\omega: f(p,\omega) < f(q, \omega) \}
$$
where $Q$ is a countable dense subset of $\RR^n$.
\end{proof}

Finally, we include a useful measurable version of the Bolzano--Weierstrass theorem.  An elementary
proof is found in the paper of Kabanov \& Stricker \cite{KS}.
 
\begin{proposition}\label{th:select}  Let $(\xi_k)_{k \ge 0}$ be a sequence of measurable functions $\xi_k: \Omega \to \RR^n$
such that $\sup_k \| \xi_k(\omega) \| < \infty$ for all $\omega \in \Omega$.  Then there exists
an increasing sequence of integer-valued measurable functions $N_k$ and an $\RR^n$-valued
 measurable function $\xi^*$ such that
$$
\xi_{N_k(\omega)}(\omega) \to \xi^*(\omega) \mbox{ as } k \to \infty
$$
for all $\omega \in \Omega$.
\end{proposition}

\section{Acknowledgement}
This work was supported
in part by the Cambridge Endowment for Research in Finance. 
I would like to thank Erhan Bayraktar for his comments and suggestions.

\end{document}